\documentclass[a4paper]{amsart}
\usepackage[foot]{amsaddr}
\usepackage[T1]{fontenc}
\usepackage{fullpage}
\usepackage{amsmath,amsthm,amssymb,amscd,amsfonts}
\usepackage{mathtools}
\usepackage[dvipsnames]{xcolor}
\definecolor{darkgreen}{rgb}{0,0.5,0}
\definecolor{darkblue}{rgb}{0,0,0.5}
\usepackage{cite}
\usepackage[colorlinks=true,urlcolor=darkblue,citecolor=darkgreen,linkcolor=darkblue]{hyperref}
\usepackage{graphicx}
\usepackage{lmodern}
\usepackage{tikz}
\usetikzlibrary{arrows.meta,calc}
\PassOptionsToPackage{hyphens}{url}

\newcommand{\bbF}{\mathbb{F}}
\newcommand{\bbZ}{\mathbb{Z}}
\newcommand{\GLn}{\mathrm{GL}_N(\bbF_2)}
\newcommand{\Spn}{\mathrm{Sp}_{2N}(\bbF_2)}

\newcommand{\Span}{\operatorname{span}}
\newcommand{\pmin}{p_{\mathrm{min}}}
\newcommand{\pmax}{p_{\mathrm{max}}}
\newcommand{\qmin}{q_{\mathrm{min}}}
\newcommand{\qmax}{q_{\mathrm{max}}}
\newcommand{\enc}{\operatorname{enc}}
\newcommand{\cost}{\operatorname{cost}}
  
\usepackage[sc]{mathpazo}
\newtheorem{thm}{Theorem}[section]
\newtheorem{prop}[thm]{Proposition}

\newtheorem{lem}[thm]{Lemma}

\theoremstyle{remark}
\newtheorem{remx}[thm]{Remark}

\theoremstyle{definition}
\newtheorem{defn}[thm]{Definition}

\newcommand{\qedblack}{\hfill \scalebox{0.7}{\ensuremath{\blacksquare}}}

\begin{document}

\title{Clifford synthesis is NP-complete}

\author{S\o ren Fuglede J\o rgensen}
  \address{Kvantify, Rosenvængets Allé 25, DK-2100 Copenhagen, Denmark}
  \email{sfj@kvantify.dk}

\date{\today}

\begin{abstract}
  We show that the problem of synthesizing a quantum circuit for a given Clifford operator using at most a specified amount of resources is NP-complete for several resource models of practical interest.

  Specifically, we study synthesis of a given Clifford operator using single-qubit Clifford gates and two-qubit Clifford gates, with full qubit connectivity and no additional ancilla qubits. With single-qubit gates free and CNOT as the only two-qubit gate, we show that deciding whether a given target admits an implementation using at most a specified number of CNOT gates is NP-complete, and we show that the problem remains NP-complete when any two-qubit Clifford gate is allowed and the budget counts all such gates. Moreover, we show that both results continue to hold when SWAP gates may be inserted for free, and when single-qubit gates are also considered part of the budget.

  In each case, we build on recent work by Li and Guan and show that there are polynomial-time reductions from Hamiltonian path in grid graphs to these decision problems, using a number of qubits quadratic in the number of grid vertices.
\end{abstract}

\maketitle


\section{Introduction}
Clifford operators play a central role in quantum error correction \cite{Gottesman1997} and magic state distillation \cite{Bravyi2005,Knill2005}, making efficient Clifford circuit synthesis an important task in quantum computing. It is also useful in quantum circuit optimization in which peephole methods replace Clifford subcircuits by equivalent circuits with fewer gates or smaller depth, reducing the resource requirements of the larger circuit \cite{Kliuchnikov2013}.

As such, a significant amount of work has been devoted to developing algorithms for Clifford circuit synthesis. These range from efficient algorithms that can synthesize instances with large qubit count \cite{Aaronson2004} but may produce circuits with high gate counts or depths, to methods for producing lookup tables giving optimal circuits for small numbers of qubits \cite{Bravyi2022}, and exact and heuristic synthesis algorithms covering intermediate regimes aiming to minimize computational requirements of the resulting circuits, generally measured in terms of circuit gate count and circuit depth \cite{Dehaene2003, Kliuchnikov2013,Amy2013, Schaeffer2014, Maslov2018,Bravyi2021,Bravyi2021b,Schneider2022,Maslov2022,Maslov2023,Peham2023, Bao2024, Shaik2025,Webster2025, Brugiere2025,Joergensen2026B}. A similar body of literature exists for the closely related problem of minimizing circuits that do not allow general Clifford gates, but only CNOT gates, and for variations such as those that allow CNOT and $T$ gates \cite{Kazuo2002, Patel2008, Meuli2018, Amy2018, Brugiere2020, Brugiere2021, Bataille2022, Gheorghiu2023, Shaik2024, Bu2025, Christensen2025,Joergensen2026, Gong2026, Acuaviva2026,Li2026}. Indeed, our work follows a line of earlier work on the complexity of circuit optimization \cite{Amy2018,Kang2023, Kjelstrom2025, Acuaviva2026, Li2026}, and builds on the recent work of Li and Guan \cite{Li2026} which showed NP-completeness of CNOT circuit synthesis.

For contemporary quantum hardware, it is generally the case that single-qubit gates are available at lower cost and with higher fidelity compared to two-qubit entangling gates. This motivates the study of Clifford circuit synthesis under cost models that are aimed at minimizing the use of two-qubit gates, and our focus is on minimizing the use of CNOT gates. By the stabilizer formalism, any Clifford operator defines a symplectic matrix $F \in \Spn$ acting on the $2N$-dimensional symplectic vector space over $\bbF_2$, and the goal is to write $F$ as a product of single-qubit Clifford gates and CNOT gates using as few CNOT gates as possible.

Our main result is that deciding whether a given Clifford operator can be synthesized using at most a specified number of CNOT gates, with single-qubit gates being free, is NP-complete. Throughout, we will allow full connectivity of the qubits, and we will not consider additional ancilla qubits. Specifically, assume that $F$ is the symplectic matrix of a Clifford operator. If $C$ is a circuit composed of single-qubit Clifford gates and CNOT gates which implements $F$, let $\cost(C)$ be the number of CNOT gates in this circuit, and let $\cost(F)$ denote the minimal cost among all circuits that implement $F$. Of particular interest will also be linear reversible operators, which are Clifford operators that can be implemented using only CNOT gates. Any such operator may be represented by an invertible binary matrix $A \in \GLn$ and defines a Clifford operator with symplectic matrix
\[
  \rho(A) = \begin{pmatrix} A & 0 \\ 0 & (A^{-1})^\top \end{pmatrix}
\]
acting on the same number of qubits.
\begin{thm}
  \label{thm:main}
  Deciding whether $\cost(F) \leq k$ is NP-complete. NP-completeness holds already for targets $F = \rho(A)$ with $A$ lower unitriangular and $(A - I)^3 = 0$.
\end{thm}
Amy et al.\ \cite{Amy2018} studied parity networks, namely CNOT circuits that generate a prescribed set of parities, possibly at different stages of the computation. They proved that deciding whether such a network exists within a given gate budget is NP-complete when all CNOT gates share a common target qubit. Their reduction encodes a hypercube walk visiting a prescribed set of vertices. Li and Guan \cite{Li2026} show NP-completeness for CNOT-minimization of linear reversible operators by using a tight gate budget to force a Hamiltonian path in the induced subgraph on a prescribed set of vertices: the reduction is from Hamiltonian path in induced subgraphs of the hypercube, closely related to the Hamiltonian path problem in grid graphs for which NP-completeness is known \cite{Itai1982}. Any path through such a subgraph has the property that two consecutive vertices differ by a single bit flip, corresponding to a CNOT operation. In turn, the composition of these operations defines a linear operator. To ensure that any circuit for such an operator has the desired structure, they further introduce \emph{recorder qubits}; whenever a new parity is encountered, its value is copied to a number of recorder qubits, helping ensure that the circuit uses the intended sequence of CNOT operations, and that no shortcuts are taken. They show that $m$ recorder qubits per parity suffice to enforce the intended structure, where $m$ is the number of vertices in the subgraph. The final linear operator is one acting on $m^2 + d + 1$ qubits. See Figure~\ref{fig:reduction-overview} for an example.

In our case, we have the further obstacle that shortcuts through the use of single-qubit Clifford gates must also be prevented. We find that if we build on the construction of Li and Guan and extend the number of recorder qubits for each parity from $m$ to $6(m - 1) + 1$, the resulting linear operator has the property that single-qubit Clifford gates cannot be used to bypass the intended sequence of CNOT operations. That is, any circuit implementing the operator with a CNOT gate count below some carefully chosen threshold can be transformed into a CNOT-only circuit without increasing the CNOT gate count. As such, we also end up with a polynomial-time reduction from the Hamiltonian path problem in a grid graph with specified endpoints.

While preparing this manuscript, we became aware of the recent work of Keller and Kueng \cite{Keller2026} which shows NP-completeness of deciding whether Clifford circuits can be implemented below a given depth for a given gate set. The work at hand is concerned with cost functions that involve gate count, and not depth, and is therefore complementary to the work in \cite{Keller2026}.

In Section~\ref{sec:background}, we introduce the necessary background on Clifford synthesis and the symplectic representation of Clifford operators.

In Sections~\ref{sec:bounding-recorder-counts} and \ref{sec:data-recorder-variables}, we provide a characterization of CNOT gates in Clifford circuits into separate types, based on their action on the symplectic vector space. In a CNOT-only circuit, all gates correspond to gates of one type (called $T$ below), and we introduce a series of bounds that ultimately allow us to ensure that any circuit implementing within the specific CNOT budget the Clifford synthesis instance defined by a hypercube Hamiltonian path instance as above uses CNOT gates only of that one type.

In Section~\ref{sec:reduction}, we describe and prove the reduction from grid Hamiltonian path to Clifford synthesis. Readers interested primarily in the main ideas of the reduction may wish to start there.

In Section~\ref{sec:arbitrary-two-qubit-clifford}, we provide generalizations to other cost models of practical relevance. We show that NP-completeness extends to the case where the resource cost includes arbitrary two-qubit Clifford gates, rather than just CNOT gates, and to the case where single-qubit Clifford gates are also counted with a given fixed cost per gate. And we show that NP-completeness still holds in the model where SWAP gates are considered free; a cost model that is commonly used when optimizing quantum circuits for architectures with complete qubit connectivity.

\begin{figure}[p]
  \centering
  \resizebox{\textwidth}{!}{
\begin{tikzpicture}[
  x=1cm,y=1cm,
  line cap=round,line join=round,
  font=\small,
  overview vertex/.style={circle,fill=darkblue,draw=darkblue,inner sep=0pt,minimum size=4pt},
  overview unused/.style={circle,fill=black!25,draw=none,inner sep=0pt,minimum size=3pt},
  overview path/.style={draw=darkblue,line width=1.15pt,-{Latex[length=1.7mm,width=1.2mm]}},
  overview wire/.style={draw=black!70,line width=.4pt},
  overview control/.style={circle,fill=darkblue,inner sep=0pt,minimum size=3.5pt},
  overview recording control/.style={circle,fill=darkgreen,inner sep=0pt,minimum size=3.5pt},
  overview target/.style args={#1}{
    circle,draw=#1,fill=white,line width=.65pt,
    inner sep=0pt,outer sep=0pt,minimum size=6pt,
    path picture={
      \draw[#1,line width=.65pt]
        (path picture bounding box.west) -- (path picture bounding box.east)
        (path picture bounding box.south) -- (path picture bounding box.north);
    }
  },
  overview target/.default=darkblue,
  overview recording target/.style={overview target=darkgreen},
  overview title/.style={anchor=west,font=\small\bfseries},
  overview note/.style={font=\scriptsize,text=black!70,align=center}
]
  \node[overview title] at (0,.22) {(a) Grid instance};
  \foreach \xx in {.6,1.8,3.0}
    \draw[black!13,densely dashed] (\xx,-3.25) -- (\xx,-1.55);
  \foreach \yy in {-3.0,-1.8}
    \draw[black!13,densely dashed] (.35,\yy) -- (3.25,\yy);
  \coordinate (g0) at (.6,-3.0);
  \coordinate (g1) at (1.8,-3.0);
  \coordinate (g2) at (1.8,-1.8);
  \coordinate (g3) at (3.0,-1.8);
  \draw[overview path] (g0) -- (g1);
  \draw[overview path] (g1) -- (g2);
  \draw[overview path] (g2) -- (g3);
  \foreach \v in {0,1,2,3} \node[overview vertex] at (g\v) {};
  \draw[darkblue,line width=.55pt] (g0) circle[radius=3.7pt];
  \draw[darkblue,line width=.55pt] (g3) circle[radius=3.7pt];
  \node[font=\scriptsize,anchor=north] at ($(g0)+(0,-.12)$) {$(0,0)$};
  \node[font=\scriptsize,anchor=north] at ($(g1)+(0,-.12)$) {$(1,0)$};
  \node[font=\scriptsize,anchor=south] at ($(g2)+(0,.12)$) {$(1,1)$};
  \node[font=\scriptsize,anchor=south] at ($(g3)+(0,.12)$) {$(2,1)$};
  \node[font=\small,text=darkblue,anchor=east] at ($(g0)+(-.16,0)$) {$s$};
  \node[font=\small,text=darkblue,anchor=west] at ($(g3)+(.16,0)$) {$t$};

  \node[overview title] at (4.15,.22) {(b) Hypercube embedding};
  \coordinate (h000) at (4.75,-3.05);
  \coordinate (h100) at (6.45,-3.05);
  \coordinate (h010) at (5.60,-2.25);
  \coordinate (h110) at (7.30,-2.25);
  \coordinate (h001) at (4.75,-1.40);
  \coordinate (h101) at (6.45,-1.40);
  \coordinate (h011) at (5.60,-.60);
  \coordinate (h111) at (7.30,-.60);
  \foreach \u/\v in {000/100,000/010,000/001,100/110,100/101,010/110,010/011,001/101,001/011,110/111,101/111,011/111}
    \draw[black!23,line width=.45pt] (h\u) -- (h\v);
  \foreach \v in {010,110,001,011} \node[overview unused] at (h\v) {};
  \draw[overview path] (h000) -- (h100);
  \draw[overview path] (h100) -- (h101);
  \draw[overview path] (h101) -- (h111);
  \foreach \v in {000,100,101,111} \node[overview vertex] at (h\v) {};
  \draw[darkblue,line width=.55pt] (h000) circle[radius=3.7pt];
  \draw[darkblue,line width=.55pt] (h111) circle[radius=3.7pt];
  \node[font=\scriptsize,anchor=east,fill=white,inner sep=1pt] at ($(h000)+(-.13,0)$) {$000$};
  \node[font=\scriptsize,anchor=north] at ($(h100)+(0,-.12)$) {$100$};
  \node[font=\scriptsize,anchor=west,fill=white,inner sep=1pt] at ($(h101)+(.20,-.08)$) {$101$};
  \node[font=\scriptsize,anchor=south] at ($(h111)+(0,.12)$) {$111$};
  \node[font=\scriptsize,text=darkblue,anchor=north] at (5.60,-3.13) {$x_1$};
  \node[font=\scriptsize,text=darkblue,anchor=west,fill=white,inner sep=1pt] at (6.55,-2.32) {$x_3$};
  \node[font=\scriptsize,text=darkblue,anchor=south east,fill=white,inner sep=1pt] at (6.77,-.93) {$x_2$};

  \node[overview title] at (8.55,.22) {(c) Circuit from path};
  \foreach \jj/\yy in {1/-.75,2/-1.30,3/-1.85} {
    \draw[overview wire] (9.65,\yy) -- (16.45,\yy);
    \node[anchor=east,font=\scriptsize] at (9.48,\yy) {$x_{\jj}$};
  }
  \draw[overview wire,line width=.6pt] (9.65,-2.70) -- (16.45,-2.70);
  \node[anchor=east,font=\scriptsize] at (9.48,-2.70) {$x_0$};
  \foreach \xx/\yy in {10.60/-.75,11.90/-1.85,13.85/-1.30} {
    \draw[darkblue,line width=.75pt] (\xx,\yy) -- (\xx,-2.70);
    \node[overview control] at (\xx,\yy) {};
    \node[overview target] at (\xx,-2.70) {};
  }
  \node[font=\scriptsize,anchor=north] at (10.10,-2.94) {$x_0$};
  \node[font=\scriptsize,anchor=north] at (11.25,-2.94) {$x_0+x_1$};
  \node[font=\scriptsize,anchor=north] at (12.875,-2.94) {$x_0+x_1+x_3$};
  \node[font=\scriptsize,anchor=north] at (15.15,-2.94) {$x_0+x_1+x_2+x_3$};

  \draw[black!20,line width=.4pt] (0,-4.08) -- (16.55,-4.08);

  \node[overview title] at (0,-4.58) {(d) Full CNOT circuit with recorder qubits};
  \foreach \cc in {3.0,5.9,8.8,11.7}
    \fill[black!2] ({\cc-1.16},-5.48) rectangle ({\cc+1.16},-12.37);
  \node[font=\scriptsize,anchor=south] at (3.0,-5.39) {$f_{000}=x_0$};
  \node[font=\scriptsize,anchor=south] at (5.9,-5.39) {$f_{100}=x_0+x_1$};
  \node[font=\scriptsize,anchor=south] at (8.8,-5.39) {$f_{101}=x_0+x_1+x_3$};
  \node[font=\scriptsize,anchor=south] at (11.7,-5.39) {$f_{111}=x_0+x_1+x_2+x_3$};
  \foreach \jj/\yy in {1/-5.90,2/-6.35,3/-6.80} {
    \draw[overview wire] (1.60,\yy) -- (13.40,\yy);
    \node[anchor=east,font=\scriptsize] at (1.40,\yy) {$x_{\jj}$};
    \node[anchor=west,font=\scriptsize] at (13.62,\yy) {$x_{\jj}$};
  }
  \draw[overview wire,line width=.6pt] (1.60,-7.55) -- (13.40,-7.55);
  \node[anchor=east,font=\scriptsize] at (1.40,-7.55) {$x_0$};
  \node[anchor=west,font=\scriptsize] at (13.62,-7.55) {$f_{111}$};
  \foreach \cc/\bits/\base in {3.0/000/-8.25,5.9/100/-9.35,8.8/101/-10.45,11.7/111/-11.55} {
    \foreach \aa/\offset in {1/0,2/-.20,r/-.65} {
      \pgfmathsetmacro{\yy}{\base+\offset}
      \draw[overview wire] (1.60,\yy) -- (13.40,\yy);
    }
    \draw[black!55,line width=.45pt]
      (1.49,\base) -- (1.37,\base) -- (1.37,{\base-.65}) -- (1.49,{\base-.65});
    \node[anchor=east,font=\scriptsize] at (1.18,{\base-.325}) {$y_{\bits,a}$};
    \foreach \dy in {-.34,-.425,-.51} {
      \fill[black!65] (1.80,{\base+\dy}) circle[radius=.35pt];
      \fill[black!65] (13.20,{\base+\dy}) circle[radius=.35pt];
    }
    \draw[black!55,line width=.45pt]
      (13.51,\base) -- (13.62,\base) -- (13.62,{\base-.65}) -- (13.51,{\base-.65});
    \node[anchor=west,font=\scriptsize] at (13.78,{\base-.325}) {$y_{\bits,a}+f_{\bits}$};
    \foreach \dx/\offset in {-.70/0,-.35/-.20,.70/-.65} {
      \pgfmathsetmacro{\xx}{\cc+\dx}
      \pgfmathsetmacro{\yy}{\base+\offset}
      \draw[darkgreen,line width=.65pt] (\xx,-7.55) -- (\xx,\yy);
      \node[overview recording control] at (\xx,-7.55) {};
      \node[overview recording target] at (\xx,\yy) {};
    }
    \node[font=\scriptsize,fill=black!2,inner sep=1.2pt,text=darkgreen] at ({\cc+.17},-7.55) {$\cdots$};
  }
  \foreach \xx/\yy in {4.45/-5.90,7.35/-6.80,10.25/-6.35} {
    \draw[darkblue,line width=.75pt] (\xx,\yy) -- (\xx,-7.55);
    \node[overview control] at (\xx,\yy) {};
    \node[overview target] at (\xx,-7.55) {};
  }
\end{tikzpicture}}
  \caption{Overview of the reduction on an example grid graph with $m = 4$ vertices and endpoints $s=(0,0)$ and $t=(2,1)$.
    (a) A Hamiltonian path visits $(0,0)$, $(1,0)$, $(1,1)$, and $(2,1)$ in that order.
    (b) The encoding maps this path to $000\to100\to101\to111$ in the three-dimensional hypercube; the other cube vertices and edges are shown in grey.
    (c) Three CNOTs, controlled by coordinate qubits $x_1$, $x_3$, and $x_2$ in that order, update the path qubit; the parities on its successive wire segments identify the Hamiltonian path.
    (d) The full circuit shows the coordinate wires and $r$ recorder qubits for each vertex. Only the first two and last recorder wires and their CNOTs are drawn in each group. Blue CNOTs move the path qubit along the hypercube edges, and green CNOTs add its current parity $f_v$ to each recorder $y_{v,a}$. Here, $m = 4$, so $r = 6(m - 1) + 1 = 19$ and the circuit uses $4r + 3 = 79$ CNOT gates. The main goal of the paper is to show that any other Clifford circuit implementing the same operator as this circuit, using no more than $R+m-1$ CNOT gates must still be such that the Hamiltonian path can be read off from the parities on the path qubit $x_0$.}
  \label{fig:reduction-overview}
\end{figure}
\clearpage

\section{Background and notation}
\label{sec:background}
Let $\bbF_2 = \{0, 1\}$ be the field with two elements. The $N$-qubit Pauli operators are length $N$ tensor products of the Pauli matrices $I$, $X$, $Y$, and $Z$ up to scalar phases. A Clifford operator is a unitary operator whose action by conjugation maps Pauli operators to Pauli operators.  Throughout, Clifford operators are considered equivalent if they differ by an overall scalar phase. This convention leaves their conjugation action on Pauli operators unchanged.

Let $X_i$ and $Z_i$ denote the $2N$ Pauli operators in which $X$ or $Z$ acts on qubit $i$, i.e., the tensor products whose $i$th factor is $X$ or $Z$ and all other factors are $I$. Conjugation by an $N$-qubit Clifford operator maps $X_i$ and $Z_i$ to a signed tensor product of the matrices $I$, $X$, $Y$ and $Z$. Discarding this sign, an $N$-qubit Clifford operator is represented by a symplectic matrix $F \in \Spn$, also often referred to as a symplectic tableau. Recording the signs gives an additional \emph{phase vector} $p \in \bbF_2^{2N}$, and the full signed tableau $(F, p)$ determines the Clifford unitary up to overall scalar phase. The phase vector can be modified arbitrarily without changing $F$ using only single-qubit Clifford gates, so when such gates are considered free, we may ignore the phase vector entirely and treat a Clifford synthesis instance as an element $F \in \Spn$. For the more general results in Section~\ref{sec:arbitrary-two-qubit-clifford}, we allow single-qubit Clifford gates to have positive cost but in this case, we find that our lower-bound arguments nevertheless only depend on the symplectic part.

The vector space $\bbF_2^{2N}$ on which the operators act will be indexed by $X_1, \dots, X_N, Z_1, \dots, Z_N$, and the symplectic form $\omega$ is defined such that $\omega(X_i, Z_j) = \delta_{ij}$ and $\omega(X_i, X_j) = \omega(Z_i, Z_j) = 0$ for all $i, j$.

For a given symplectic matrix $F \in \Spn$, we will write
\[
  a_i = F_{X_i}, \quad b_i = F_{Z_i}
\]
for the rows belonging to qubit $i$. Appending a gate $G$ to a circuit $U$ corresponds to the left-action, so the corresponding matrices transform as $F_{GU} = F_G F_U$. Thus, for instance, a CNOT gate with control $i$ and target $j$ acts by
\[
 a_j\leftarrow a_j+a_i,\qquad b_i\leftarrow b_i+b_j.
\]
In particular, a CNOT-only circuit implementing a linear transformation $A \in \GLn$ has symplectic matrix
\[
  \rho(A) = \begin{pmatrix} A & 0 \\ 0 & (A^{-1})^\top \end{pmatrix}.
\]
For a given $F \in \Spn$ and a given qubit $i$, we write
\[
  P_i = \mathrm{span} \{a_i, b_i\}.
\]
These planes are nondegenerate and pairwise symplectically orthogonal, and their direct sum is the full space $\bbF_2^{2N}$. A single-qubit Clifford gate changes the basis $a_i, b_i$ of $P_i$ but preserves $P_i$ itself; it will be convenient to work with the planes $P_i$ rather than the individual basis vectors. Here, a subspace $W \subseteq \bbF_2^{2N}$ is \emph{nondegenerate} if $W \cap W^\perp = \{0\}$, where
\[
  W^\perp = \{ v \in \bbF_2^{2N} \mid \omega(v, w) = 0 \text{ for all } w \in W \}.
\]
A subspace $W$ is \emph{isotropic} if $W \subseteq W^\perp$ and \emph{Lagrangian} if $W = W^\perp$. Recall also that an isotropic and a Lagrangian subspace have dimensions at most $N$ and exactly $N$, respectively.

Much of our work amounts to tracking how the symplectically orthogonal planes $P_i = \mathrm{span} \{a_i, b_i\}$ intersect with the Lagrangian subspace $W$, and how these intersections evolve under the action of Clifford gates. That is, for a given circuit, we really have a sequence of tableaux, thus also a sequence of planes $P_i$, but to avoid cluttering the notation, we will avoid explicitly indexing the sequence.

One example of a Lagrangian subspace is the subspace spanned by the first $N$ (or last $N$) standard basis vectors of $\bbF_2^{2N}$, but in what follows, many of our results will be independent of the specific choice of $W$.

\section{Characterizing CNOT gate actions}
\label{sec:bounding-recorder-counts}
Let $(\bbF_2^{2N}, \omega)$ be the symplectic vector space defined above, and let $W$ denote a general Lagrangian subspace of $\bbF_2^{2N}$. Let $F \in \Spn$ be a symplectic tableau, and let $a_i = F_{X_i}$ and $b_i = F_{Z_i}$ for $i = 1, \dots, N$ denote its rows, and let $P_i = \mathrm{span} \{a_i, b_i\}$.

Note immediately that since any single-qubit Clifford gate leaves all $P_i$ invariant, it in particular does not change $W \cap P_i$. 
\begin{defn}
  We say that qubit $i$ \emph{represents} a vector $w \in W \setminus \{0\}$ if $w \in P_i$. An \emph{appearance} of such a vector on qubit $i$ is an occurrence in which the qubit represents the vector immediately after a CNOT gate when the same vector was not represented by qubit $i$ immediately before the gate.
\end{defn}
Since $W$ is isotropic and $P_i$ is nondegenerate, $\dim(W \cap P_i) \leq 1$ for all $i$, and thus each qubit can represent at most one nonzero vector from $W$ at any given time.

We will group the action of CNOT gates into six distinct types based on how they affect the intersections $W \cap P_i$. For a CNOT with control and target $i$ and $j$ at some point in the circuit, let $v_-, v_+ \in \{0, 1, 2\}$ count how many of these two qubits represent a nonzero vector of $W$ immediately before and immediately after the gate, respectively, and let $k \leq v_+$ be the number of appearances caused by the gate. The six types are then characterized by the values of the triple $(v_-, v_+, k)$, and we name them as follows:
\begin{itemize}
  \item \makebox[8em][l]{Type $T$ (top):}
    $(v_-, v_+, k) = (2, 2, 1)$.
  \item \makebox[8em][l]{Type $D$ (down):}
    $(v_-, v_+, k) = (2, 0, 0)$.
  \item \makebox[8em][l]{Type $U$ (up):}
    $(v_-, v_+, k) = (0, 2, 2)$.
  \item \makebox[8em][l]{Type $P$ (plus):}
    $(v_-, v_+, k) = (0, 1, 1)$.
  \item \makebox[8em][l]{Type $M$ (minus):}
    $(v_-, v_+, k) = (1, 0, 0)$.
  \item \makebox[8em][l]{Type $Z$ (zero):}
    $(v_-, v_+, k) \in \{(0, 0, 0), (1, 1, 0), (2, 2, 0)\}$.
\end{itemize}
\begin{lem}
  \label{lem:six-types}
  The six types of CNOT gate actions described above exhaust all possible changes to the intersections $W \cap P_i$. For type $T$, the represented vectors change from $(u, v)$ to either $(u, u+v)$ or $(u + v, v)$.
\end{lem}
\begin{proof}
  Suppose that the CNOT gate acts on qubits $i$ and $j$, and write the planes before the application of the gate as $P_i = \Span \{ a, b \}$ and $P_j = \Span \{ c, d \}$. After application of the gate, these planes become
  \[
    P_i' = \Span \{ a, b + d \}, \quad P_j' = \Span \{ c + a, d \}.
  \]
  The four-dimensional space $S = \Span \{ a, b, c, d \}$ is left unchanged. Let $K = W \cap S$. Recall that we have $\omega(a, b) = \omega(c, d) = 1$ and that the symplectic form $\omega$ vanishes on pairs of vectors from different qubits, so $\{a, b, c, d\}$ is a symplectic basis of $S$, and the form $\omega$ restricted to $S$ is nondegenerate. Since $W$ is Lagrangian, $K$ is isotropic and therefore has dimension at most $2$.

  If $\dim K = 0$, none of the vectors $a, b, c, d$ belong to $W$, and therefore the intersections $W \cap P_i$ and $W \cap P_j$ are trivial both before and after the application of the gate. Thus the gate has type $Z$.
  
  Suppose that $\dim K = 1$, and let $w \in K$ be its nonzero vector. At most one qubit can represent $w$ both before and after the gate, so $v_-, v_+, k \in \{0, 1\}$. Moreover, it is impossible for $w$ to be represented by one qubit before the gate, and the other qubit after, as
  \[
    \Span \{ a, b \}  \cap \Span \{c + a, d \} = 0, \quad \Span \{ c, d \} \cap \Span \{a, b + d \} = 0,
  \]
  which rules out the case $(v_-, v_+, k) = (1, 1, 1)$. Thus, the CNOT is of type $P$, $M$, or $Z$.

  Now suppose that $\dim K = 2$. First, note that $v_-, v_+ \in \{0, 2\}$, i.e. that if one qubit represents a vector, then the other does too, and this is the case both before and after the gate. To see this, suppose one qubit represents a vector $u$. Since its plane lies in $S$, we have $u \in W \cap S = K$. Choose $w \in K$ independent of $u$ and decompose $w = p + q$ over the two qubit planes where $p$ lies in the plane containing $u$. Since $W$ is isotropic and the two qubit planes are symplectically orthogonal,
  \[
    0 = \omega(w, u) = \omega(p, u) + \omega(q, u) = \omega(p, u),
  \]
  so $p$ is a multiple of $u$, and therefore $q = w - p \in K \setminus \{0\}$ is represented by the other qubit.

  If both qubits represent vectors before the gate, write those as $u \in \{a, b, a+b\}$ and $v \in \{c, d, c+d\}$. Then $K = \Span \{u, v\}$. To find the type of the CNOT gate, it suffices to check which of $u, v, u+v$ belong to $\Span \{a, b+d \}$ and $\Span \{c + a, d\}$. Comparing coefficients in $a, b, c, d$ gives the following four cases:
  \begin{itemize}
    \item If $u = a$ and $v = d$, both vectors remain represented by their respective qubits.
    \item If $u = a$ and $v \not=d$, the represented vectors after the gate are $u$ and $v + u$, and $v + u$ is an appearance.
    \item If $u \not= a$ and $v = d$, the represented vectors are $u+v$ and $v$.
    \item If $u \not= a$ and $v \not= d$, neither of the two qubits represents a vector after the gate.
  \end{itemize}
  The four cases correspond to gates of type $Z, T, T, D$ respectively. In particular, this rules out the case $(v_-, v_+, k) = (2, 2, 2)$, showing that the six types cover all possibilities.
\end{proof}
Next, we will bound the total number of appearances that can occur in a given circuit in terms of the number of gates of the above types, the goal being to show that for the circuits relevant for the hardness construction, only gates of type $T$ can occur.

Consider a given Clifford circuit consisting of only single-qubit Clifford gates and $L$ CNOT gates. Each CNOT gate can cause up to $2$ appearances. Let $B$ denote the sum of the number of appearances over all CNOT gates.

In a slight abuse of notation, let $T, D, U, P, M, Z$ denote the number of CNOT gates of the corresponding types as defined above.
\begin{lem}
  \label{lem:first-bounds}
  Suppose that the number of qubits representing a vector of $W$ is the same initially and after the circuit. Then
  \[
    B \leq L, \quad L - B = Z + \frac{P + M}{2}, \quad P + M + Z \leq 2(L - B).
  \]
\end{lem}
\begin{proof}
  By Lemma~\ref{lem:six-types},
  \[
    B = T + 2U + P, \quad L = T + D + U + P + M + Z.
  \]
  As the number of qubits representing a vector of $W$ is assumed to be unchanged by the circuit,
  \[
    2U + P = 2D + M
  \]
  and therefore
  \[
    L - B = D - U + M + Z = Z + \frac{P+M}{2}
  \]
  and
  \[
    \frac{P + M + Z}{2} \leq L - B. \qedhere
  \]
\end{proof}
Next, we will bound the number of appearances that can occur outside type $T$ and type $U$ gates. First, recall a general fact about symplectic vector spaces.
\begin{lem}
  Let $S$ be a nondegenerate symplectic subspace. If $W \cap S$ is Lagrangian in $S$, then
  \[
    W = (W \cap S) \oplus (W \cap S^\perp).
  \]
\end{lem}
Suppose a type $D$ gate acts on qubits $i, j$, which represent vectors $u$ and $v$ before the gate. As in the proof of Lemma~\ref{lem:six-types}, the space $S = P_i \oplus P_j$ is unchanged by the gate. The vectors $u$ and $v$ are independent, so
\begin{align}
\label{eq:WcapS}
  W \cap S = \Span \{u, v\}, \quad W = (W \cap S) \oplus (W \cap S^\perp).
\end{align}
After the gate, neither of the two qubits represents a vector in $W$ until the next CNOT gate that involves either $i$ or $j$ and the decomposition above remains valid. In the following, we track the history of such unrepresented qubits.

To track appearances coming from different types of gates, maintain a list of pairs of qubits, and a set $A$. Each pair consists of two qubits that represent no vector and have not participated in a CNOT since the type $D$ gate creating the pair, while the set $A$ contains all other qubits that represent no vector.

Concretely, assume that every qubit initially represents a vector of $W$ (which is the case, for instance, if $W = \Span \{x_1, \dots, x_N\}$). Start with an empty list of pairs and an empty set $A$. Then, update them as follows:
\begin{itemize}
  \item At a type $D$ gate, add its two qubits as a pair.
  \item At a gate of type $P$, $M$, or $Z$, remove every listed pair containing a qubit involved in the gate. Remove from $A$ any participating qubit that represents a vector after the gate. For each participating qubit, and each partner from a removed pair, add it to $A$ if it represents no vector after the gate.
  \item At a type $U$ gate, remove the two participating qubits from the list of pairs, or from $A$. The following lemma shows that they are either one listed pair or both in $A$. Let $U_1$ and $U_2$ be the two subcases corresponding to these possibilities.
  \item Type $T$ gates do not affect the list of pairs or the set $A$.
\end{itemize}
\begin{lem}
  \label{lem:easy-gates}
  Suppose that every qubit initially represents a vector of $W$. Every type $U$ gate acts either on one listed pair (type $U_1$) or on two members of $A$ (type $U_2$), and every type $P$ gate acts on two members of $A$.

  Writing $U_2$ also for the number of gates of the corresponding type, the number of appearances at gates other than those of types $T$ and $U_1$ is
  \[
    P + 2U_2 \leq 3M + 4Z.
  \]
  If every qubit also finally represents a vector of $W$, then
  \[
    P + 2U_2 \leq 3M + 4Z \leq 6(L-B).
  \]
\end{lem}
\begin{proof}
  Suppose that $\{i, j\}$ is a listed pair and that $k \not= j$ is another qubit that represents no vector. We claim that
  \[
    W \cap (P_i \oplus P_k) = 0.
  \]
  Indeed, let $w \in W \cap (P_i \oplus P_k)$ and write $w = w_i + w_k$ with $w_i \in P_i$, $w_k \in P_k$. By~\eqref{eq:WcapS}, $w_i \in W$ and $w_k \in W$, but $W \cap P_i = W \cap P_k = 0$ since the qubits represent no vectors so $w = 0$. The space $P_i \oplus P_k$ is preserved by a CNOT acting on $i$ and $k$; if $P_i'$ and $P_k'$ are the planes after such a gate, we therefore also have $W \cap P_i' = \{0\}$ and $W \cap P_k' = \{0\}$, so the qubits continue to represent no vectors. In particular, such a CNOT gate must have type $Z$.

  With this, we can track the behavior of listed pairs and the set $A$ through the circuit:
  \begin{itemize}
  \item A type $U$ gate involving a member of a listed pair must also involve its partner in that pair.
  
  \item A type $P$ gate cannot involve both members of a listed pair at all since, as we have seen in the proof of Lemma~\ref{lem:six-types}, when $\dim (W \cap (P_i \oplus P_j)) = 2$, if one qubit represents a vector, then the other does too. Therefore, a type $P$ gate can only act by removing a single qubit from $A$.

  \item A type $T$ or type $D$ gate cannot involve a listed qubit since both participating qubits represent vectors before such a gate.

  \item A type $Z$ gate can add at most four qubits to the set $A$: its control and target as well as possibly one partner from a listed pair for each.

  \item A type $M$ gate can add at most three qubits to the set $A$: its two participating qubits and, if the previously unrepresented participant belonged to a listed pair, its former partner as well.
  \end{itemize}

  Let $I_A$ be the total number of additions to the set $A$ throughout the circuit, allowing multiple additions of the same qubit, let $U_2$ count the number of type $U$ gates acting on two members of $A$, and let $R_A$ count all other removals. Then
  \[
    \lvert A \rvert = I_A - 2U_2 - R_A \geq 0
  \]
  and
  \[
    2U_2 + P \leq 2U_2 + R_A \leq I_A  \leq 4Z + 3M.
  \]
  If, moreover, every qubit also finally represents a vector of $W$, then by Lemma~\ref{lem:first-bounds}, we have
  \[
    4Z + 3M \leq 3P+3M+6Z \leq 6(L-B). \qedhere
  \]
\end{proof}

\section{Data and recorder variables}
\label{sec:data-recorder-variables}
Suppose now that the $N$ qubits are split into two sets: $d + 1 \geq 1$ ``data'' qubits, followed by $R \geq 0$ ``recorder'' qubits. Denote the first $N$ rows of the initial tableau by $x_0, \dots, x_d, y_1, \dots, y_R$. These are the first $N$ standard basis vectors in $\bbF_2^{2N}$. We call $x_0, \dots, x_d$ \emph{data variables} and $y_1, \dots, y_R$ \emph{recorder variables}, and we say that a \emph{data vector} is an element of $\Span \{x_0, \dots, x_d\}$. Let
\[
  W = \Span \{x_0, \dots, x_d, y_1, \dots, y_R\}
\]
be the Lagrangian subspace associated with these variables.

Let $f_r \in \Span \{x_0, \dots, x_d\} \setminus \{0\}$, $r = 1, \dots, R$, be nonzero data vectors. Assume that all final qubits represent vectors of $W$ and that, moreover, recorder qubit $r$, which initially represents $y_r$, is required to represent $q_r = y_r + f_r$ at the end of a given circuit. Note that the $q_r$ are pairwise distinct, and they are neither initial basis vectors nor data vectors.

Recorder qubit $r$ does not initially represent $q_r$, so there must be one or more CNOT gates that cause it to appear. Let the \emph{marked appearance} be the final such appearance.

There are $R$ marked appearances, one for each recorder qubit. Let
\[
  E = B - R
\]
be the number of appearances that are not marked.

Next, we will bound how many outputs of the form $q_r = y_r + f$ can possibly be realized by the recorder qubits if $f$ itself is never represented on any qubit.
\begin{lem}
  \label{lem:recorder-count}
  Let $f$ be a nonzero data vector. If no qubit ever represents $f$, then the number of recorder qubits with $f_r = f$ satisfies
  \[
    \lvert \{ r \mid f_r = f \}\rvert \leq 2E + 3M + 4Z \leq 6(L-R).
  \]
\end{lem}
\begin{proof}
  By Lemma~\ref{lem:easy-gates}, at most $3M + 4Z$ appearances, and hence at most $3M + 4Z$ marked appearances, can occur outside type $T$ gates and type $U_1$ gates. We will count how many recorder qubits $r$ can have $f_r = f$ and a marked appearance of one of those two types.

  First, consider those recorder qubits $r$ for which there was an earlier appearance of any nonzero vector in $W$ on the same recorder qubit before the marked appearance. Each such earlier appearance is not marked, so there are at most $E$ such recorder qubits.

  For every remaining recorder qubit $r$, the only vector represented on its qubit before the marked appearance is the initial vector $y_r$.
  
  Suppose its marked appearance occurs at a type $T$ gate. Then, by Lemma~\ref{lem:six-types}, the appearance of $q_r = y_r + f$ at this gate requires that $f$ was also represented before the type $T$ gate, contradicting the assumption that $f$ is never represented on any qubit.

  Suppose instead that the marked appearance occurs at a type $U_1$ gate on a listed pair. At the type $D$ gate that created this pair, the represented vectors were $y_r$ on the recorder qubit and some nonzero vector $b$ on the other qubit of the pair. Call this qubit $j$. By~\eqref{eq:WcapS},
  \[
    q_r = y_r + f \in \Span \{y_r, b\}.
  \]
  Since $q_r \not= y_r$, either $b = q_r$ or $b = f$; the latter contradicts the assumption that $f$ is never represented on any qubit, so $b = q_r$. In particular, qubit $j$ represents $q_r$ immediately before the type $D$ gate. Since the $\{q_s\}_{s=1}^R$ are pairwise distinct, the appearance of $q_r$ on qubit $j$ is not a marked appearance. Associating this appearance to qubit $r$, as above, this case contributes at most another $E$ recorder qubits. The associations may overlap but provide an overall bound of $2E$, proving the first inequality.

  Since all qubits represent vectors initially and finally, Lemma~\ref{lem:easy-gates} and the fact that $E = B-R \geq 0$ imply that
  \[
    2E + 3M + 4Z \leq 2E + 6(L - B) = 6(L-R) - 4E \leq 6(L-R). \qedhere
  \]
\end{proof}
By Lemma~\ref{lem:first-bounds}, if $B = L$, there can be no gates of types $P$, $M$, or $Z$. The following lemma gives a condition for also excluding gates of type $D$ and type $U$.
\begin{lem}
  \label{lem:all-T}
  Suppose that every qubit represents a vector of $W$ initially and finally, and that $B = L$. Suppose also that every appearance is of a vector that has not previously been represented on any qubit, including initially. Then every CNOT gate has type $T$.
\end{lem}
\begin{proof}
  As noted above, Lemma~\ref{lem:first-bounds} rules out gates of type $P$, $M$, and $Z$, so we have to consider only gates of type $D$, $U$, and $T$. Recall the construction, preceding Lemma~\ref{lem:easy-gates}, of the set $A$ of certain qubits representing no vectors. This set is initially empty, and no gates of types $D$, $U$, or $T$ add elements to $A$. Therefore, every type $D$ gate creates a listed pair of qubits representing no vectors; as gates of type $T$ require both participating qubits to represent vectors immediately before the gate, the next gate acting on the qubits in a listed pair must be of type $U$. Such a later gate must exist since both qubits must represent vectors at the end of the circuit.

  Let $u, v$ be the vectors before a gate of type $D$. The later type $U$ gate produces two independent vectors in $\Span \{u, v\}$. Both count as appearances since neither qubit represents a vector immediately before the gate. But the only nonzero vectors in this span are $u$, $v$, and $u+v$, and both $u$ and $v$ have previously been represented. Thus at most one of the two output vectors can be a vector that has not been represented before, contradicting the assumption. Thus there are no gates of type $D$.

  Since all qubits initially represent vectors and there are no gates that make them cease to do so, there are also no type $U$ gates.
\end{proof}

\section{The reduction}
\label{sec:reduction}
Recall that we aim to reduce the problem of finding a minimum CNOT circuit for a given Clifford operation from Hamiltonian path in a grid graph with specified endpoints for which NP-completeness is known \cite{Itai1982}. Here, a grid graph $G$ is the subgraph of the integer square lattice $\bbZ^2$ in which vertices correspond to points with integer coordinates and edges connect any pair of vertices at unit distance.

First, suppose the graph $G$ is connected, has $m \geq 2$ vertices, and has distinct specified endpoints $s_G, t_G \in V(G)$. A Hamiltonian path in $G$ is a path that visits each vertex exactly once, and we require that it starts at $s_G$ and ends at $t_G$. Let
\begin{gather*}
  \pmin = \min \{ p \mid (p, q) \in V(G) \}, \quad 
  \pmax = \max \{ p \mid (p, q) \in V(G) \}, \\
  \qmin = \min \{ q \mid (p, q) \in V(G) \}, \quad
  \qmax = \max \{ q \mid (p, q) \in V(G) \},
\end{gather*}
and
\[
  d = (\pmax - \pmin) + (\qmax - \qmin).
\]
Following Li and Guan \cite{Li2026}, define an injection $\enc : V(G) \to \bbF_2^d$ by
\[
  \enc(p, q) = (1, \dots, 1, 0, \dots, 0, 1, \dots, 1, 0, \dots, 0),
\]
where the numbers of $1$s and $0$s in each block are $p - \pmin$, $\pmax - p$, $q - \qmin$, and $\qmax - q$, respectively. Note that two vertices are adjacent in $G$ if and only if their encodings differ in exactly one position. Define
\[
  V = \{\enc(v) + \enc(s_G) \mid v \in V(G) \}, \quad t = \enc(t_G) + \enc(s_G).
\]
Then $\lvert V \rvert = m$, and we have $0, t \in V$ and $t \not= 0$. The subgraph $Q_d[V]$ of the $d$-dimensional hypercube whose vertices are $V$, and whose edges connect pairs of vertices that differ in exactly one position, is isomorphic to $G$ with endpoints $0$ and $t$. Recall also from \cite{Li2026} that $d \leq m - 1$.

\subsection{The target operator}
\label{sec:target-operator}
Let $r = 6(m-1) + 1$, $R = mr$, and $N = d + 1 + R$. As in Section~\ref{sec:data-recorder-variables}, introduce data variables $x_0, x_1, \dots, x_d$, and for each $v \in V$ introduce $r$ recorder variables $y_{v,a}$, $a = 1, \dots, r$. For each $v = (v_1, \dots, v_d) \in V$, define $f_v \in \bbF_2^{2N}$ by
\[
  f_v = x_0 + \sum_{j=1}^d v_j x_j.
\]
We will require that our target operator maps $x_0$ to $f_t$, that it maps $x_j$ to itself for $j = 1, \dots, d$, and that it maps the variable $y_{v,a}$ to $y_{v,a} + f_v$, for all $v \in V$ and $a = 1, \dots, r$. We refer to the qubit corresponding to $x_0$ as the \emph{path qubit}.

Let $A \in \GLn$ be the matrix of this operator in the basis $x_1, \dots, x_d, x_0, y_{v,a}$, where $v \in V$ and $a = 1, \dots, r$. That is,
\[
  A = \begin{pmatrix}
    I_d & 0 & 0 \\
    t & 1 & 0 \\
    X & \mathbf{1} & I_R \\
  \end{pmatrix},
\]
where $X$ contains $r$ copies of each vertex row, in recorder order. Note that $A$ is lower unitriangular, that $A - I$ is strictly lower triangular and $(A-I)^3 = 0$.

The output instance of the reduction is the matrix $F = \rho(A) \in \Spn$ and the gate count threshold $k = R + m - 1$. Since $d \leq m-1$, the number of qubits satisfies $N=O(m^2)$. The dense binary tableau $F$ therefore has $4N^2=O(m^4)$ entries. Both $F$ and $k$ are computable in polynomial time; the quadratic bound refers to the number of qubits.

Note that compared to \cite{Li2026}, we have used $r \geq m$ copies of each recorder qubit instead of only $m$. For the CNOT-only case, the exact choice does not affect the reduction, and by replicating the argument, one finds that any number of copies at least $m$ would work.
\begin{lem}
  \label{lem:li-guan-generalised}
  Let $s \geq m$ be an integer and let $R = ms$ and $N = d + 1 + R$.  Let $X_s \in \bbF_2^{R\times d}$ contain $s$ copies of each vertex of $V$, in any fixed order, and define
  \[
  A_s=
  \begin{pmatrix}
    I_d & 0 & 0 \\
    t & 1 & 0 \\
    X_s & \mathbf{1} & I_R
  \end{pmatrix}.
  \]
  Then $A_s$ admits a circuit consisting only of CNOTs, on exactly $N$ labelled qubits with unrestricted connectivity, using at most $R + m - 1$ gates, if and only if $Q_d[V]$ has a Hamiltonian path from $0$ to $t$.
\end{lem}
\begin{proof}
  As in \cite{Li2026}, a Hamiltonian path gives a CNOT circuit by updating the path qubit along its $m-1$ edges and, at each vertex, adding its parity to the corresponding $s$ recorder qubits; see also Figure~\ref{fig:reduction-overview}. This uses $ms + m - 1 = R + m - 1$ CNOT gates.

  Conversely, suppose a CNOT circuit implements $A_s$ using $L \leq R + m - 1$ gates. Every recorder qubit must be targeted at least once, since its prescribed output differs from its input. If some vertex parity $f_v$ is never represented, each of the $s$ recorder qubits corresponding to $v$ must be targeted at least twice. Thus, for this vertex, the $s$ recorder qubits require at least $2s$ CNOT gates, and the remaining $R - s$ recorder qubits require at least one CNOT each, so
  \[
    L \geq (R-s) + 2s = R + s \geq R + m > R+m-1,
  \]
  which is a contradiction. Thus every vertex parity is represented. The remaining argument of \cite{Li2026} applies unchanged.
\end{proof}

\begin{prop}
  If $Q_d[V]$ has a Hamiltonian path from $0$ to $t$, then $\cost(F) \leq k$.
\end{prop}
\begin{proof}
  By Lemma~\ref{lem:li-guan-generalised}, there exists a CNOT circuit implementing $A_r$ with at most $R + m - 1 = k$ gates. The same circuit, viewed as a Clifford circuit, implements $F = \rho(A)$, so $\cost(F) \leq k$.
\end{proof}

\begin{prop}
  \label{prop:soundness}
  If a Clifford circuit implements $F = \rho(A)$ with $L \leq R + m - 1$ CNOT gates, then $Q_d[V]$ has a Hamiltonian path from $0$ to $t$.
\end{prop}
\begin{proof}
  Let $W$ be the span of all initial first $N$ rows. Since the target is $\rho(A)$, each final plane $P_i$ meets $W$ in the span of the corresponding row of $A$. In particular, all qubits represent vectors initially and finally, and Lemma~\ref{lem:recorder-count} applies.

  Suppose that there exists a $v \in V$ such that $f_v$ is not represented at any point of the circuit. Then Lemma~\ref{lem:recorder-count}, applied to the $r$ recorders corresponding to $v$, would give
  \[
    6(m-1) + 1 = r \leq 6(L-R) \leq 6(m - 1),
  \]
  so each $f_v$ must be represented at some point in the circuit.

  There are $R$ recorder qubits and hence $R$ marked appearances. Each of the $m - 1$ vectors $f_v$, $v \not= 0$, also requires an appearance; these vertex parities are pairwise distinct and differ from all recorder outputs, so
  \[
    R + m - 1 \leq B \leq L \leq R + m - 1.
  \]
  It follows that the only appearances are the $R$ marked recorder appearances and one appearance of each $f_v$ for $v \not= 0$. None of the vectors is an initial vector and no vector appears twice, so by Lemma~\ref{lem:all-T}, every CNOT has type $T$.

  At every stage, take the unique represented vector from each qubit. These are linearly independent because they belong to distinct summands $P_i$. As they are $N$ vectors in the $N$-dimensional space $W$, they form a basis of $W$. By Lemma~\ref{lem:six-types}, relative to the initial basis, the corresponding row matrix starts at the identity, changes by one elementary row addition at each CNOT, and ends at $A$. We have therefore obtained a CNOT-only circuit for $A$ with exactly $R + m - 1$ gates.

  By Lemma~\ref{lem:li-guan-generalised}, this implies that $Q_d[V]$ has a Hamiltonian path from $0$ to $t$.
\end{proof}

\begin{proof}[Proof of Theorem~\ref{thm:main}]
  The construction of the Clifford synthesis instance $F$ is polynomial in the size of the Hamiltonian path instance, and the two preceding propositions show that a grid graph has a Hamiltonian path with specific endpoints if and only if  $\cost(F) \leq k$. This proves NP-hardness and by the form of the target operator, we see that it holds already for instances $F = \rho(A)$ with $A$ unitriangular and $(A - I)^3 = 0$.

  For membership in NP, combine consecutive single-qubit gates into one layer between each pair of CNOT gates, and into layers at the input and output. A circuit with $L$ CNOT gates then has at most $L + 1$ single-qubit layers, and each layer consists of at most $N$ elements of the fixed finite group $\mathrm{Sp}_2(\bbF_2)$, so it has a description of size $O(N)$.

  There is a polynomial upper bound $p(N)$ on the number of CNOT gates needed to synthesize an arbitrary Clifford operator \cite{Aaronson2004}. For a yes-instance with $k < p(N)$, a circuit with at most $k$ CNOT gates is a polynomial-size certificate. If $k \geq p(N)$, use a synthesis with at most $p(N)$ CNOT gates instead. Multiplying the gate tableaux and comparing the result with $F$ verifies the certificate in polynomial time.

  In the construction, we have assumed that $m \geq 2$, that the grid graph is connected, and that the endpoints are distinct. An instance with a disconnected graph, or one with equal endpoints and $m \geq 2$, has no Hamiltonian path and can be sent to a fixed no-instance. The case $m = 1$ with both endpoints at that single vertex is a yes-instance and can be sent to the one-qubit identity with threshold zero. 
\end{proof}

\section{Arbitrary two-qubit Clifford gates}
\label{sec:arbitrary-two-qubit-clifford}
To prove that NP-completeness extends to the case in which not only CNOT gates but arbitrary two-qubit Clifford gates are considered as part of the cost, we will express a circuit of arbitrary two-qubit Clifford gates as a circuit of CNOT gates and single-qubit Clifford gates followed by a permutation. We then extend the reduction to allow this output permutation.

\begin{lem}
  \label{lem:two-qubit-clifford-classification}
  Every two-qubit Clifford gate is a product of single-qubit Clifford gates, at most one CNOT, and at most one SWAP. Consequently, a circuit with $L$ arbitrary two-qubit Clifford gates and free single-qubit Clifford gates has a symplectic representation of the form
  \[
    F = \rho(\Pi) F_C,
  \]
  where $\Pi$ is a permutation matrix and $C$ is a circuit with at most $L$ CNOT gates and any number of single-qubit Clifford gates.
\end{lem}
\begin{proof}
  It is known (see e.g. \cite{Corcoles2013}, \cite[Thm.~26]{Grier2022}) that every two-qubit Clifford gate is equivalent, up to single-qubit Clifford gates, to an empty circuit, to a single CNOT gate, to an iSWAP gate (which, in turn, is equivalent to the composition of a CNOT gate and a SWAP gate), or to a SWAP gate. In particular, we can choose a decomposition using at most one CNOT gate for every two-qubit gate. Moreover, any SWAP gate can be moved past a later gate by relabelling the qubits, which gives the desired decomposition.
\end{proof}
The permutation may cause recorder outputs to be moved to different qubits; to account for this, we need a new version of Lemma~\ref{lem:recorder-count}. Recall the setup of Section~\ref{sec:data-recorder-variables} and the split of the qubits into data and recorder qubits denoted $x_0, \dots, x_d$ and $y_1, \dots, y_R$ respectively and let $f_r \in \Span \{x_0, \dots, x_d\} \setminus \{0\}$, $r = 1, \dots, R$, be nonzero data vectors. Let $W$ be the span of the initial first $N$ rows, and suppose that at the end of a circuit with $L$ CNOT gates, every qubit represents a vector of $W$, with the vectors $q_r = y_r + f_r$, $r = 1, \dots, R$, being represented on distinct qubits. This time, we take a marked appearance of $q_r$ to be the appearance that begins its final uninterrupted occurrence on its final qubit and as before, let $E = B - R$ be the number of appearances that are not marked.
\begin{lem}
  \label{lem:recorder-count-permuted}
  If a nonzero data vector $f$ is never represented, then
  \[
    \lvert \{r \mid f_r = f \} \rvert \leq 3E + 3M + 4Z \leq 6(L - R).
  \]
\end{lem}
\begin{proof}
  Repeat the partition in the proof of Lemma~\ref{lem:recorder-count}, using the final qubit of $q_r$ in place of the original recorder qubit $r$. As before, by Lemma~\ref{lem:easy-gates}, at most $3M + 4Z$ marked appearances occur outside type $T$ and type $U_1$ gates, and the association with earlier appearances on the same qubit still accounts for at most $E$ further outputs, since the final qubits are distinct. As such, we need only modify the proof for the remaining gate types for qubits for which there are no appearances preceding the marked ones.

  As such, fix an $r$ with $f_r = f$, and let $e_r$ be the initial basis vector on the final qubit representing $q_r$, so unlike before, $e_r$ is not necessarily $y_r$, and assume that that qubit has represented only $e_r$ before the marked appearance. Denote by $j$ the other qubit involved in the gate causing the marked appearance. The same type $T$ and $U_1$ arguments show that another qubit earlier represented a vector $b_r$ satisfying
  \[
    b_r = q_r + e_r \quad \text{in type }T, \qquad b_r \in \{q_r, q_r + e_r\} \quad \text{in type } U_1.
  \]
  If $b_r = q_r$, then as in Lemma~\ref{lem:recorder-count}, the appearance of $q_r$ on qubit $j$ is not a marked appearance.
  
  Suppose that $b_r = q_r + e_r$. If $e_r = y_r$, then $b_r = f$, contradicting the assumption that $f$ is never represented. If $e_r$ is a data basis vector, the projections of $b_r$ in the data basis and recorder basis are $f + e_r$ and $y_r$, respectively. Here, $f + e_r$ is nonzero because otherwise $f$ would be equal to $e_r$, contradicting the assumption that $f$ is never represented. If $e_r = y_s$, $s \not= r$, then $b_r = y_r + y_s + f$ has two nonzero coordinates in the recorder basis. Thus in either case, $b_r$ is neither an initial vector, nor any recorder output, so its earlier appearance is not a marked appearance.

  Associate one such unmarked appearance to each output not covered by the first $3M + 4Z + E$ associations. The associations with $b_r = q_r$ are distinct and do not coincide with those of the form $b_s = q_s + e_s$ since the latter are not recorder outputs. If distinct $r, s$ in the latter case give the same vector, $q_r + e_r = q_s + e_s$, then $y_r + y_s = e_r + e_s$. Since $e_r$ and $e_s$ are basis vectors, this gives $\{e_r, e_s \} = \{y_r, y_s\}$, and since $e_r \not= y_r$, we have $e_r = y_s$ and $e_s = y_r$, so $b_r = q_r + e_r = y_r + f + y_s$, so $b_r$ has two nonzero coordinates in the recorder basis. Thus, each unmarked appearance is associated to at most two outputs, contributing at most $2E$ outputs.

  The total number is at most $3E + 3M + 4Z$, and as in Lemma~\ref{lem:recorder-count},
  \[
    3E + 3M + 4Z \leq 3E + 6(L-B) = 6(L - R) - 3E \leq 6(L - R). \qedhere
  \]
\end{proof}
Recall the constructions of the target operator $A$ and gate count threshold $k$ from Section~\ref{sec:target-operator}. We now generalize Proposition~\ref{prop:soundness} by showing that allowing an output permutation does not change the circuit structure forced by the budget $k$.
\begin{prop}
  \label{prop:soundness-permutation}
  If a circuit with at most $k$ CNOT gates implements $\rho(\Pi A)$ for a permutation matrix $\Pi$, then $Q_d[V]$ has a Hamiltonian path from $0$ to $t$.
\end{prop}
\begin{proof}
  Every final plane represents a vector, and the final represented vectors are the rows of $A$ in some order. By Lemma~\ref{lem:recorder-count-permuted} with $r = 6(m - 1) + 1$, every $f_v$ is represented. As in Proposition~\ref{prop:soundness}, the only appearances are the $R$ marked appearances and one appearance of each $f_v$ for $v \not= 0$. By Lemma~\ref{lem:all-T} every CNOT gate has type $T$.

  Consider the row matrix formed by the unique nonzero vectors in $W \cap P_i$. Every row belongs to the list consisting of initial basis vectors, the vertex parities $f_v$, and the recorder outputs $q_{v,a} = y_{v,a} + f_v$. Each vector in this list that is not initial is produced by a row addition exactly once. In particular, once a vector is no longer represented, it cannot be produced later.

  We will show that each coordinate and recorder output remains on its prescribed row, i.e., $\Pi$ is the trivial permutation. A coordinate vector $x_j$, $j = 1, \dots, d$, occurs initially and is required among the final rows. No row addition can produce it, so its original row is never updated.

  Next, consider the row addition that produces $q_r$. Its two operands must be $y_r$ and $f_r$. To see this, note again that at every stage of the circuit, every row is either an initial basis vector, a vertex parity, or a prescribed recorder output; and in particular, every row has at most one nonzero coordinate in the recorder basis. Since $q_r$ has recorder part $y_r$, exactly one operand has this recorder part, and the other has recorder part $0$. The only possible rows with recorder part $y_r$ are $y_r$ itself and $q_r$. Since this is the first appearance of $q_r$, that operand must be $y_r$, and the other must be $q_r + y_r = f_r$.

  Consider the first update of the row initially equal to $y_r$. There must be such an update, because $y_r$ is not a final row. This update produces either a recorder output $q_s$ or a vertex parity $f_v$. If it produces $q_s$, by the argument above, $s = r$. If it produces $f_v$, the control of the gate must have recorder part $y_r$ to cancel the recorder part of $y_r$. The control cannot be $y_r$ itself because the rows are independent, so it must be $q_r$. Thus, the new row is $y_r + q_r = f_r$, but by the above, producing $q_r$ requires that $f_r$ is already available as an operand, so this operation would produce a previously represented vector, which is impossible.

  Thus, the first update of the original recorder row $y_r$ must give $q_r$. Since $q_r$ is required finally and cannot be produced again, this row is never updated afterwards. It follows that all coordinate and recorder rows finish in their original positions with their prescribed values.

  The only remaining row is the path row. This row must finish at $f_t$, the circuit implements $A$ itself, and hence $\Pi$ is trivial. The result now follows from Lemma~\ref{lem:li-guan-generalised}.
\end{proof}
We say that two two-qubit gates are \emph{locally Clifford-equivalent} if they are equal up to single-qubit Clifford gates.
\begin{thm}
  Fix a set $\Gamma$ of two-qubit Clifford gates containing at least one gate locally Clifford-equivalent to the CNOT gate. Deciding whether a target Clifford operator can be implemented using at most $k$ gates from $\Gamma$ and free single-qubit Clifford gates is NP-complete.

  This result holds for the targets $\rho(A)$ with $A$ lower unitriangular and $(A - I)^3 = 0$.

  In particular, this result applies to the case where $\Gamma$ consists of all two-qubit Clifford gates.
\end{thm}
\begin{proof}
  Given an instance of Hamiltonian path in grid graphs, construct the target $F = \rho(A)$ and let $k = R + m - 1$ as in Section~\ref{sec:target-operator}.

  If a Hamiltonian path exists, there is a CNOT circuit for $F$ with $k$ CNOT gates. Replace each gate with a gate from $\Gamma$ that is locally Clifford-equivalent to it. This gives a circuit for $F$ with $k$ gates from $\Gamma$.

  Conversely, suppose a circuit for $\rho(A)$ uses $L \leq k$ gates from $\Gamma$. By Lemma~\ref{lem:two-qubit-clifford-classification}, this gives a circuit with at most $L$ CNOTs implementing $\rho(\Pi^{-1})\rho(A) = \rho(\Pi^{-1}A)$, and Proposition~\ref{prop:soundness-permutation} gives a Hamiltonian path.

  The certificate and tableau verification argument from Theorem~\ref{thm:main} applies here as well, ensuring membership in NP.
\end{proof}
The above also proves NP-completeness of CNOT count (and the more general $\Gamma$ count) when SWAP gates are free: move all SWAP gates to the end of the circuit and apply Proposition~\ref{prop:soundness-permutation}.

The above machinery is also sufficient to prove NP-completeness in a cost model where single-qubit Clifford gates are available at some more general cost. Fix integers $a \geq 0$ and $b > 0$. Assign cost $a$ to each single-qubit Clifford gate and cost $b$ to each gate from a fixed set $\Gamma$ of two-qubit Clifford gates containing CNOT. With respect to these fixed costs, for a Clifford operator $U$, let $\cost_{\Gamma,a,b}(U)$ denote the minimal cost of any circuit implementing $U$. In this case, we can no longer ignore the phase vector of the Clifford operator: for example, if $a > 0$, an operator like the single-qubit NOT gate whose symplectic part is trivial can no longer be implemented at cost zero; nevertheless, we find that we may take the phase vector to be zero in the reduction.
\begin{thm}
  Let $a \geq 0$ and $b > 0$ be integers, and let $\Gamma$ be a set of two-qubit Clifford gates containing CNOT. Given a Clifford operator $U$, specified by its full signed tableau, deciding whether $\cost_{\Gamma,a,b}(U) \leq k$ is NP-complete. The result holds both with and without additional free SWAP gates, and NP-completeness holds already for targets $(\rho(A), 0) \in \Spn \times \bbF_2^{2N}$ with $A$ lower unitriangular and $(A - I)^3 = 0$.
\end{thm}
\begin{proof}
  Use the symplectic matrix $F = \rho(A)$ and $k = R + m - 1$ as in the previous reductions, and let the threshold be $K = bk$, but take now the target to be the Clifford operator whose full signed tableau is $U_A = (\rho(A), 0) \in \Spn \times \bbF_2^{2N}$. A Hamiltonian path gives a CNOT-only circuit with $k$ gates. This circuit has cost $K$ and uses no single-qubit gates.

  Conversely, consider any circuit implementing $U_A$ with $S$ single-qubit gates and $L$ gates from $\Gamma$ with total cost $aS + bL$. Passing to its symplectic representation gives a circuit implementing $\rho(A)$, and if the cost is at most $K$, then $L \leq k$, so we may ignore the single-qubit costs and apply Lemma~\ref{lem:two-qubit-clifford-classification}. The resulting circuit has at most $L$ CNOT gates and possibly an output permutation, so Proposition~\ref{prop:soundness-permutation} gives the required Hamiltonian path. If free SWAP gates are allowed, first move them to the output by relabelling subsequent gates; then the same argument applies.

  For membership in NP, there is a polynomial upper bound $q(N)$ on the total cost of synthesizing any signed Clifford tableau since $\Gamma$ contains CNOT. If $K \geq q(N)$, then that synthesis supplies a polynomial-size certificate. Otherwise, the circuit has $L \leq q(N)/b$ two-qubit gates from $\Gamma$. Combining local gates between them leaves at most $N(L+1)$ single-qubit gates without increasing the cost. Note also that the full signed tableau may still be verified in polynomial time. Free SWAP gates can be moved to a final permutation represented by at most $N - 1$ SWAP gates.
\end{proof}
In the above, we assumed that $a$ and $b$ were integers; by multiplying both costs and the budget by a common denominator, the result extends to rational costs as well.

\subsection*{Acknowledgements}
This work was funded by the Innovation Fund Denmark (Grand Solutions) -- grant no. 5366-00005B.
\subsection*{Use of large language models}

Large language models were used during the preparation of the paper, in particular to prepare drafts of some of the proofs. We found that current frontier models did not reach the main results without further context but that given the construction of Li and Guan \cite{Li2026}, they were able to identify the characterization of CNOT gate actions in Clifford circuits, as well as the counting arguments used to ensure that single-qubit Clifford gates would not lead to shortcuts for the instances constructed from the Hamiltonian path problem. The resulting proof drafts required substantial digestion and verification, and the proofs presented here were rewritten from scratch by the author. Models were also used to provide preliminary auto-formalization, proofreading, and review.

\bibliographystyle{unsrturl}
\bibliography{references}

\end{document}